\documentclass[aps,amsmath,onecolumn,amssymb,groupedaddress,showpacs,nofootinbib,notitlepage,longbibliography]{revtex4-1}

\usepackage[utf8]{inputenc} 
\usepackage{amsmath}
\usepackage{amsfonts}
\usepackage{amsthm}

\usepackage{graphicx}
\usepackage{booktabs} 
\usepackage{array} 
\usepackage{paralist} 
\usepackage{verbatim} 
\usepackage{algorithmic}
\usepackage{algorithm}
\usepackage{color}
\usepackage{hyperref}
\usepackage{qcircuit}

\theoremstyle{plain}
\newtheorem{theorem}{Theorem}[section]
\newtheorem{lemma}[theorem]{Lemma}
\newtheorem{proposition}[theorem]{Proposition}

\theoremstyle{definition}
\newtheorem{definition}[theorem]{Definition}
\newtheorem{example}[theorem]{Example}

\newcommand{\ket}[1]{|#1\rangle}

\newcommand{\F}{\mathbb{Z}_2}
\newcommand{\Z}{\mathbb{Z}}
\newcommand{\x}{\mathbf{x}}
\newcommand{\y}{\mathbf{y}}
\newcommand{\z}{\mathbf{z}}
\newcommand{\w}{\mathbf{w}}
\renewcommand{\a}{\mathbf{a}}
\renewcommand{\b}{\mathbf{b}}
\renewcommand{\c}{\mathbf{c}}
\newcommand{\f}{\mathbf{f}}
\renewcommand{\P}{\mathcal{P}}
\newcommand{\C}{\mathcal{C}_n}
\newcommand{\R}{\mathcal{RM}}
\newcommand{\B}{\mathcal{B}}

\newcommand{\CNOT}{\mathrm{CNOT}}
\newcommand{\cnot}{\mathrm{CNOT}}
\newcommand{\wt}[1]{\lvert#1\rvert}
\newcommand{\res}{\mathrm{Res}}

\renewcommand{\deg}{\mathrm{deg}}
\newcommand{\diag}{\mathrm{diag}}
\newcommand{\0}{\mathbf{0}}
\renewcommand{\1}{\mathbf{1}}

\newcommand{\fvar}{X}
\newcommand{\dist}[2]{\delta(#1, #2)}

\newcommand{\etal}{{\it et al. }}

\newcommand{\ev}[1]{\mathbf{#1}}

\newcommand{\FF}[1]{\F^{#1}\setminus\{\0\}}
\newcommand{\FO}[1]{\F^{#1}\setminus\{\1\}}

\newcommand{\poly}{P}
\newcommand{\mono}[3]{#1_1^{#2_1}#1_2^{#2_2}\cdots#1_#3^{#2_#3}}
\newcommand{\xor}[3]{#1_1#2_1\oplus#1_2#2_2\oplus\cdots\oplus#1_#3#2_#3}

\AtBeginDocument{
\heavyrulewidth=.08em
\lightrulewidth=.05em
\cmidrulewidth=.03em
\belowrulesep=.65ex
\belowbottomsep=0pt
\aboverulesep=.4ex
\abovetopsep=0pt
\cmidrulesep=\doublerulesep
\cmidrulekern=.5em
\defaultaddspace=.5em
}

\begin{document}

\title{T-count optimization and Reed-Muller codes}

\author{Matthew Amy}
\email{meamy@uwaterloo.ca}
\affiliation{Institute for Quantum Computing, University of Waterloo, Waterloo, ON, Canada}
\affiliation{David R. Cheriton School of Computer Science, University of Waterloo, Waterloo, ON, Canada}

\author{Michele Mosca}
\email{michele.mosca@uwaterloo.ca}
\affiliation{Institute for Quantum Computing, University of Waterloo, Waterloo, ON, Canada}
\affiliation{Department of Combinatorics \& Optimization, University of Waterloo, Waterloo, ON, Canada}
\affiliation{Perimeter Institute for Theoretical Physics, Waterloo, ON, Canada}
\affiliation{Canadian Institute for Advanced Research, Toronto, ON, Canada}

\begin{abstract}
In this paper, we study the close relationship between Reed-Muller codes and single-qubit phase gates from the perspective of $T$-count optimization. We prove that minimizing the number of $T$ gates in an $n$-qubit quantum circuit over $\CNOT$ and $T$, together with the Clifford group powers of $T$, corresponds to finding a minimum distance decoding of a length $2^n-1$ binary vector in the order $n-4$ punctured Reed-Muller code. Moreover, we show that the problems are polynomially equivalent in the length of the code. As a consequence, we derive an algorithm for the optimization of $T$-count in quantum circuits based on Reed-Muller decoders, along with a new upper bound of $O(n^2)$ on the number of $T$ gates required to implement an $n$-qubit unitary over $\CNOT$ and $T$ gates. We further generalize this result to show that minimizing small angle rotations corresponds to decoding lower order binary Reed-Muller codes. In particular, we show that minimizing the number of $R_z(2\pi/d)$ gates for any integer $d$ is equivalent to minimum distance decoding in $\R(n - k - 1, n)^*$, where $k$ is the highest power of $2$ dividing $d$.
\end{abstract}

\maketitle

\section{Introduction}

The synthesis and optimization of quantum circuits has generated a great deal of interest in recent years. As qubit technologies become more stable and experimentalists increase the size of their systems, actually running algorithms on these machines becomes a practical concern. Moreover, we want to know how to \emph{efficiently} run these algorithms on the given systems, or conversely how big and stable of a system we need to run a particular algorithm. Given the prevalence of the circuit model within quantum computing, quantum circuit optimization is an important tool in answering these questions.

Due to the great affect of noise on quantum computations, much research has shifted its focus from optimizing physical circuits to logical ones with respect to a fault-tolerance schemes meant to mitigate the errors due to this noise. These schemes usually have striking differences from physical gates in terms of resource costs. In particular, most of the common schemes implement Clifford group gates \emph{transversally} -- that is, by performing one physical gate on each physical qubit or group of qubits. This allows the logical operation to be performed precisely and with time proportional to the physical gate time. The additional operation needed to make a universal gate set is then typically implemented probabilistically with state distillation and gate teleportation, a less accurate procedure which requires both additional time and space compared to a single physical gate. The two qubit controlled-NOT ($\cnot$) gate, as an element of the Clifford group, is hence a relatively cheap operation in this paradigm, compared to the $T=\diag(1, e^{i\frac\pi4})$ gate which is commonly chosen as the non-Clifford gate. This is a reversal of the computational costs inherent in most physical implementations, where entangling gates are typically more difficult to implement than single qubit rotations, and hence requires different circuit optimizations. While alternative fault-tolerance methods such as Paetznick and Reichardt's completely transversal Clifford+$T$ scheme \cite{pr13} and anyonic quantum computing \cite{k03} are gaining in popularity, minimization of the number of $T$ gates -- called the $T$-count -- in quantum circuits remains an important and widely studied goal.

We build on previous work by Amy, Maslov and Mosca on the reduction of $T$ gates in quantum circuits. In \cite{ammr13} it was shown that unitaries implementable over $\CNOT$ and $T$ gates may be described as a (linear) permutation together with a phase rotation that is an $8$th root of unity given by a pseudo-Boolean function of the input bits in the computational basis. This function, called the circuit's \emph{phase polynomial}, was shown to be expressible as a weighted sum of linear Boolean functions, each function corresponding to the application of a $T$ gate to a power given by its weight. This idea was later used in \cite{amm14} to optimize both $T$-count and $T$-depth -- the minimal number of stages of parallel $T$-gates in a circuit -- by computing a circuit's phase polynomial, simplifying it, then synthesizing a new circuit from the polynomial with maximally parallelized $T$ gates. While their benchmarks showed significant reduction of $T$ gates, it was noted that this approach was not optimal, as it was shown that there exist distinct phase polynomials that give rise to the same unitary. In particular, it was observed that for all $x_1,x_2,x_3,x_4\in\F$,
\[e^{i\frac\pi4\sum_{f\in\F^4\rightarrow\F}f(x_1,x_2,x_3,x_4)}=1=e^0,\] where $\F^n\rightarrow\F$ is the space of all $n$-bit linear Boolean functions. It was left as an open question as to whether there exist other such identities, and whether such identities can be used to further reduce instances of $T$ gates.

In this paper we fully characterize the set of identities between phase polynomials on $n$ qubits. In doing so, we find that the set of identities on $n$ qubits that are useful for reducing a circuit's $T$-count correspond exactly to the code-words of the length $2^n-1$ punctured Reed-Muller code of order $n-4$. This allows us to derive a new $T$-count optimization algorithm based on Reed-Muller decoding which is optimal for $\cnot$ and $T$ gate circuits when a minimum distance decoder is used. We implemented this optimization algorithm as a module in the quantum circuit optimizer $T$-par \cite{tpar} and tested it on general Clifford+$T$ circuits with two different Reed-Muller decoders. The results show modest reductions in $T$-count while still remaining tractable for circuits of non-trivial size, further confirming the efficacy of the (polynomial-time) $T$-par algorithm \cite{amm14} in terms of $T$-count optimization. Our result further provides a new quadratic upper bound on the number of $T$ gates required to implement a circuit over $\{\cnot, T\}$, along with evidence towards the intractability of exact $T$-count minimization via a polynomial-time equivalence to the minimum distance decoding problem for the punctured Reed-Muller code.

Our proof naturally generalizes to the case when the $T$ gate is replaced with a $Z$ rotation by any angle of the form $2\pi/2^k$. These gate sets are closely related to the Clifford-cyclotomic gate sets studied in \cite{fgkm15}, and are widely used in quantum algorithms including Shor's algorithm \cite{s94}. We show that minimizing the number of $2\pi/2^k$ rotation gates for each value of $k$ corresponds to decoding punctured Reed-Muller codes of order $n-k-1$, opening up the possibility of optimizing such circuits at the high level before decomposing them into a lower level gate set such as Clifford+$T$. We further show that these are the \emph{only} non-trivial identities between phase polynomials over arbitrary angles -- in particular, minimizing rotation gates of any composite order $2^x3^y5^z\cdots$ reduces to the case of order $2^x$.

\subsection{Related work}
Much work has gone into $T$-count and depth reduction in recent years.
Amy \etal \cite{ammr13} identified the $T$-count and $T$-depth as important quantities in the efficiency of a logical quantum circuit, and gave new implementations of 2--4 bit quantum operations reducing $T$-count and depth from the previously best known. Their search-based algorithm was later extended by Gosset \etal \cite{gkmr14} to directly optimize $T$-depth, leading to proofs of $T$-depth minimality for various 2--4 bit circuits. Selinger \cite{s13} showed that the Toffoli gate, as well as a general class of Clifford+$T$ circuits, can be parallelized to $T$-depth 1 with sufficiently many ancillas. Constructions for adding controls to quantum gates were also given which lowered the $T$-count and depth compared to best known practices using Toffoli gates. Amy, Maslov and Mosca later used similar ideas to create an automated, polynomial-time tool for reducing and parallelizing $T$ gates called $T$-par, which uses matroid partitioning to parallelize the $T$ gates. More recently, Abdessaied, Soeken and Drechsler \cite{asd14} studied the effect of Hadamard gates on $T$-count and depth reductions, developing a tool that reduces Hadamard gates in quantum circuits leading to further $T$ gate optimizations. Maslov \cite{m15} examined Toffoli gate implementations up to relative phase and used them to develop new designs for multiple control Toffolis using fewer ancillas, $\cnot$, and in some cases $T$ gates, than standard designs. 

A great deal of work optimizing $T$-count and depth in single qubit circuits has also been done recently, with series of works on exact \cite{kmm12} and approximate \cite{kmm13, s15, rs14} minimal synthesis, as well as repeat-until-success circuits \cite{ps14, brs15}. While we instead focus on \emph{multi-qubit} circuit optimization, the single- and multi-qubit approaches are complementary as circuits may be first optimized at the level of abstract, small angle rotations before optimally decomposing such gates into sequences of $T$ and Clifford group gates. 

The relationship between Reed-Muller codes and $T$ gates has previously been studied from the perspective of fault-tolerance, with applications to the construction of quantum error correcting codes with transversal roots of $Z$ \cite{klz96, zcc11, aj14} or otherwise implementing such gates with magic state distillation \cite{bh12, ehd12, lc13}. Our work differs from the work done in the fault-tolerance community in that we are interested in the optimization of quantum circuits, rather than implementing phase gates fault-tolerantly -- hence we establish \emph{completeness} results in addition to the \emph{existence} results found in fault-tolerance research. 


\subsection{Overview}
The rest of the paper is organized as follows. Section~\ref{sec:prelim} gives definitions and notation that will be used throughout the paper. Section~\ref{sec:phase} defines the linear phase operators, details their representation as weighted sums of linear Boolean functions and synthesis. Section~\ref{sec:topt} defines an additive subgroup of $\Z_8^{2^n-1}$ whose cosets correspond to the unique linear phase operators, then characterizes its binary residue as a Reed-Muller code and gives applications. Section~\ref{sec:generalization} generalizes the result to circuits over $\cnot$ gates and phase rotations with angles of the form $2\pi/d$, and Section~\ref{sec:experiments} details the experimental evaluation of our technique.

\section{Preliminaries}\label{sec:prelim}

We assume some knowledge of quantum computing and coding theory, but provide the basic necessary definitions from both. For a complete introduction to quantum computing, the reader is referred to Nielsen \& Chuang \cite{nc00}, and for background on coding theory see MacWilliams \& Sloane \cite{ms78}.

\subsection{Quantum circuits}

We work in the circuit model of quantum computation \cite{nc00}. The state of an $n$-qubit quantum system is modelled as a unit vector in a dimension $2^n$ complex vector space. As is standard we denote the $2^n$ basis vectors of the \emph{computational} basis by $\ket{\x}$ for bit strings $\x=x_1x_2\cdots x_n\in\F^n$ -- these are called the \emph{classical} states. We denote binary vectors by boldface letters and use them interchangeably as bit strings. A general quantum state may be written as a \emph{superposition} of classical states
\[
  \ket{\psi} = \sum_{\x\in\F^n} \alpha_{\x}\ket{\x},
\]
for complex $\alpha_{\x}$ and having unit norm.

Quantum gates, analogous to classical gates, correspond to unitary matrices on some $2^m$ dimensional complex vector space. An $m$-qubit gate may be lifted to a gate on some $m$-qubit subset of an $n$-qubit system by taking its tensor product with the identity matrix on the unaffected qubits. By a quantum circuit over a particular set of gates we mean a sequential list of gates taken from the set, each with a list of qubits the gate is to be applied to. Such a circuit \emph{implements} a unitary operator on $n$ qubits, defined as the (sequential) product of each gate appropriately lifted to $n$ qubits. In this way, two distinct circuits may implement the same unitary matrix -- we call such circuits \emph{equivalent}. 

In this paper we will primarily be interested in two gates: the controlled-NOT ($\cnot:\ket{x}\ket{y}\mapsto\ket{x}\ket{x\oplus y}$ where $\oplus$ denotes addition in $\F$) and the $T$-gate ($T:\ket{x}\mapsto e^{i\frac\pi4x}\ket{x}$). These two gates, together with $S:=T^2$ and $Z:=T^4$ gates, comprise what we refer to for brevity as the $\{\cnot, T\}$ gate set. We include the $S$ and $Z$ gates in this set to distinguish them from sequences of $T$ gates which are generally much more expensive to implement in most fault-tolerance schemes. Given any power $k\in\Z_8$ of the $T$ gate, we define a minimal $T$-gate expansion by $$T^k:=Z^{k_2}S^{k_1}T^{k_0}$$ where $k_2k_1k_0$ is the binary expansion of $k$. Note that $T^8=I$ so $T^k=T^{k\mod 8}$ for all integers $k$.

The problem of optimizing quantum circuits is to find, given a circuit, an equivalent circuit minimizing some cost function. In cases where the cost function assigns some non-zero cost to a particular gate $U$ while all other gates are free we refer to the resulting optimization problem as the \emph{$U$-gate minimization} problem. In this work we primarily consider the problem of $T$-gate minimization over the $\{\CNOT, T\}$ gate set. It should be noted that, while the $\{\cnot, T\}$ gate set is not universal in the sense that not every $n$-qubit unitary can be implemented to arbitrary accuracy with a polylogarithmic number of $\cnot$ and $T$ gates, the addition of the Hadamard gate ($H:\ket{x}\mapsto\frac1{\sqrt{2}}\sum_{x'\in\F}(-1)^{xx'}\ket{x'}$) gives a universal set known as Clifford+$T$.

\subsection{Coding theory}

A length $n$ \emph{binary linear code} is a subspace $C$ of $\F^n$, where $\F$ is the unique $2$-element field $(\{0,1\},\oplus, \cdot)$ with addition ($\oplus$) and multiplication ($\cdot$) modulo $2$. The elements of $C$ are called the \emph{codewords} of $C$. Note that $\F$ is the set of Boolean values with addition corresponding to exclusive-OR and multiplication corresponding to AND. Addition and multiplication are extended to vectors component-wise -- that is, $\x\y$ is the component-wise multiple of vectors $\x$ and $\y$, as opposed to matrix multiplication.

We denote binary vectors by boldface letters e.g., $\x=x_1x_2\cdots x_n\in\F^n$, and use them interchangeably as bit strings. In particular, we denote the $n$-qubit computational basis vectors by $\ket{\x}$ where $\x$ is a binary vector/bit string. The \emph{(Hamming) weight} of a binary vector, denoted $\wt{\x}, \x\in\F^n$, is defined as the number of non-zero entries it contains, and the \emph{(Hamming) distance} between two binary vectors $\x,\y\in\F^n$ is the weight of their sum: $$\dist{\x}{\y}:=\wt{\x\oplus\y}.$$

Given a received vector $\x\oplus\boldsymbol{e}$ where $\x\in C$ and $\boldsymbol{e}\in\F^n$ is some error vector, we wish to find $\x$ -- this process is known as \emph{decoding}. In this work, we are only concerned with \emph{minimum distance decoding}, as it relates directly to $T$-count optimization.
\begin{definition}
Given a binary linear code $C$ and vector $\x\in\F^n$, a \emph{minimum distance decoding} of $\x$ in $C$ is a codeword $\y\in C$ such that for all $\z\in C$, $\dist{\x}{\y}\leq \dist{\x}{\z}$.
\end{definition}

The problem of finding a minimum distance decoding is closely related to the more general \emph{closest vector problem} over a lattice, and in fact coincides with the closest vector problem over the lattice $C$ with the Hamming weight as the norm. Minimum distance decoding is commonly studied as it reasonably approximates maximum likelihood decoding when bit flip errors are independent of one another.

We give one more definition from coding theory which will be relevant to our work: the maximum distance of any vector from a codeword, called the \emph{covering radius}.

\begin{definition}
The \emph{covering radius} of a length $n$ binary code $C$ is $$\rho(C)=\max_{\x\in\F^n}\min_{\y\in C}\dist{\x}{\y}.$$
\end{definition}

\subsection{Reed-Muller codes}

Many different presentations of the binary Reed-Muller codes (\cite{m54, r54}) are known; we use a presentation based on multivariate polynomials as it will provide a convenient setting for our work. For more details the reader is referred to \cite{ms78}.

Let $\F[\fvar_1,\fvar_2,\dots, \fvar_n]$ be the ring of polynomials in $n$ variables over $\F$. We use the symbols $\fvar_1,\fvar_2,\dots, \fvar_n$ to denote formal variables so as to differentiate them from elements of binary vectors. Given $f\in\F[\fvar_1,\fvar_2,\dots, \fvar_n]$ we define the \emph{evaluation vector} of $f$, when viewed as an $n$-ary function, to be the length $2^n-1$ vector consisting of the evaluation of $f$ at all non-zero inputs ordered lexicographically -- i.e. $$(f(10\cdots 0), f(01\cdots 0), \dots, f(11\cdots 1)).$$ We denote the evaluation vector of a polynomial function $f$ by $\f$. Since $\fvar^2=\fvar$ for all $\fvar\in\F$, we work in the quotient ring $f\in\F[\fvar_1,\fvar_2,\dots, \fvar_n]/\langle \fvar_1^2 - \fvar_1, \dots, \fvar_n^2-\fvar\rangle$ and assume polynomials are in reduced form with exponents $0$ or $1$. Identifying the variable $\fvar_i$ with the Boolean function $f(\fvar_1,\fvar_2,\dots, \fvar_n)=\fvar_i$, we denote the evaluation vector of $\fvar_i$ by $\ev{\fvar}_i$. It can be easily verified that for any Boolean polynomial $f=\bigoplus_{\y\in\F^n}\fvar_1^{y_1}\fvar_2^{y_2}\cdots \fvar_n^{y_n}$, the evaluation vector of $f$ is equal to $\bigoplus_{\y\in\F^n}\ev{\fvar}_1^{y_1}\ev{\fvar}_2^{y_2}\cdots \ev{\fvar}_n^{y_n}$ -- again, exponentiation of a Boolean vector is defined as component-wise exponentiation.

We define the \emph{total degree} of a monomial $\mono{\fvar}{y}{n}$ to be the sum of its exponents: $$\deg(\mono{\fvar}{y}{n})=\sum_{i=1}^ny_i=\wt{\y}.$$ The degree of a polynomial function $f\in\F[\fvar_1,\fvar_2,\cdots \fvar_n]$, denoted $\deg(f)$, is defined as the maximum total degree of each monomial. Table~\ref{tab:evals} illustrates the evaluation vectors of the $2^n$ monomials on $n$ variables. Note that the set of non-constant monomial evaluation vectors are linearly independent and form a basis for the space $\F^{2^n-1}$.

\begin{table}
\caption{Evaluation vectors for monomials over $n$ Boolean variables.}\label{tab:evals}
\centering
\begin{tabular}{c|c c c c c c}
 & $100\cdots 0$ \quad & $010\cdots 0$ \quad& $110\cdots 0$ \quad& $001\cdots 0$ \quad& $\cdots$ \quad& $111\cdots 1$\quad \\ \hline
 $1$ & 1 & 1 & 1 & 1 & $\cdots$ & 1 \\
 $\fvar_1$ & 1 & 0 & 1 & 0 & $\cdots$ & 1 \\
 $\fvar_2$ & 0 & 1 & 1 & 0 & $\cdots$ & 1 \\
 $\fvar_1\fvar_2$ & 0 & 0 & 1 & 0 & $\cdots$ & 1 \\
  $\fvar_3$ & 0 & 0 & 0 & 1 & $\cdots$ & 1 \\
 $\vdots$ & $\vdots$ & $\vdots$ & $\vdots$ & $\vdots$ & $\ddots$ & $\vdots$ \\
 $\fvar_1\fvar_2\cdots \fvar_n$ & 0 & 0 & 0 & 0 & $\cdots$ & 1
\end{tabular}
\end{table}

\begin{definition}\label{def:rm}
The \emph{punctured Reed-Muller code} of order $r$ and length $2^n-1$, denoted $\R(r, n)^*$, is the set of evaluation vectors for polynomials $f\in\F[\fvar_1,\fvar_2,\dots, \fvar_n]$ of degree at most $r$.
\end{definition}

The non-punctured, length $2^n$ Reed-Muller code or order $r$ is defined in a similar fashion, using evaluation vectors consisting of all $2^n$ distinct evaluations for a given polynomial function instead.

%
%
%

\section{Linear phase operators}\label{sec:phase}
In this section we introduce linear phase operators as the subset of unitaries implementable by $\{\CNOT, T\}$ which require $T$ gates. We review their representation as pseudo-Boolean functions and define the canonical $T$-count for a particular polynomial. Finally we show that a minimal $T$-count implementation of a linear phase operator corresponds to a minimal weight vector of a vector space coset.

We define $\P_8(n)$ to be the set of diagonal $2^n\times 2^n$ unitaries implementable over $\{\CNOT, T\}$ -- we restrict our attention to this subset as any circuit over $\{\CNOT, T\}$ may be decomposed into a diagonal unitary followed by a permutation implementable using only $\CNOT$ gates. Amy \etal \cite[Lemma 2]{ammr13} showed that each such unitary $U\in\P_8(n)$ has the effect of applying a pseudo-Boolean function $\poly$ to a computational basis state $\ket{\x}$, viewed as a vector $\x=x_1x_2\cdots x_n\in\F^n$, and kicking the result into the phase: $$U:\ket{\x}\mapsto e^{i\frac{\pi}{4}\poly(\x)}\ket{\x}.$$ Moreover, it was shown that the \emph{phase polynomial} $P:\F^n\rightarrow\Z_8$ necessarily has a presentation as a weighted sum of (non-zero) linear Boolean functions:
\[\poly(\x) = \sum_{\y\in\FF{n}} a_\y(\xor{y}{x}{n}),\] where the coefficients $a_\y$ are integers modulo $8$. We call the tuple $\a=(a_1,a_2,\dots, a_{2^n-1})\in\Z_8^{2^n-1}$ an \emph{implementation} of $\poly$, and conversely denote the phase polynomial defined by an element $\a$ of $\Z_8^{2^n-1}$ by $\poly_\a$. As the function $\poly$ involves both $\F$ and $\Z_8$ arithmetic, we implicitly use the natural inclusion of $\F$ in $\Z_8$ to lift the binary valued result of $\xor{y}{x}{n}$ into an integer.

We call unitaries in $\P_8(n)$ $\pi/4$ \emph{linear phase operators}, as they may be expressed as a sequence of $\pi/4$ phase rotations conditioned on linear Boolean functions of the input basis state. We drop the $\pi/4$ until Section~\ref{sec:generalization} when we generalize the result to $2\pi/{2^k}$ linear phase operators. Given a particular phase polynomial $\poly$, we denote the linear phase operator with phase polynomial $\poly$ by $U_\poly$.

\begin{example}
The doubly-controlled $Z$ gate is a $\pi/4$ linear phase operator with phase function $\poly(x_1, x_2, x_3) = 4x_1x_2x_3$. Using the identity $2\cdot xy = x + y + 7(x\oplus y)\mod 8$ \cite{s13}, the phase function may be given as the following weighted sum of linear Boolean functions: $$\poly(x_1, x_2, x_3) = x_1 + x_2 + 7(x_1\oplus x_2) + x_3 + 7(x_1\oplus x_3) + 7(x_2\oplus x_3) + (x_1\oplus x_2\oplus x_3).$$ Writing the coefficients above as a $7$-tuple over $\Z_8$ we get $(1, 1, 7, 1, 7, 7, 1)$. Note that this implementation corresponds to the following circuit, taken from \cite{amm14}. The state of each qubit after an update is shown to illustrate the relation between the state of a qubit as a Boolean function of the inputs and the application of phase gates.
$$
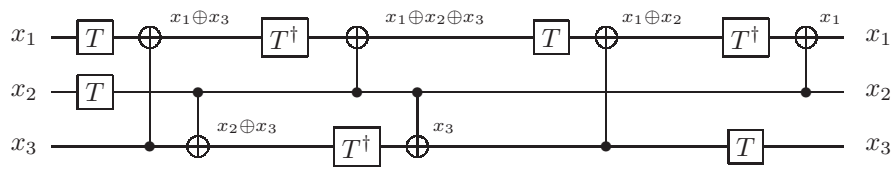
\Qcircuit @C=1em @R=.7em {
\lstick{x_1} & \gate{T} & \targ & 
\ustick{\scriptstyle\:x_1\oplus x_3}\qw & \push{\rule{0em}{1em}}\qw & 
\gate{T^\dagger} & \targ & 
\ustick{\scriptstyle\;\;\;\;\;\;x_1\oplus x_2\oplus x_3} \qw & 
\push{\rule{0em}{1em}}\qw & \push{\rule{0em}{1em}}\qw & \push{\rule{0em}{1em}}\qw &
\gate{T} & \targ & 
\ustick{\scriptstyle\;\;\;x_1\oplus x_2} \qw & \push{\rule{0em}{1em}}\qw & \push{\rule{0em}{1em}}\qw &
\gate{T^\dagger} &
\targ & \ustick{\scriptstyle\!\!\!\!\!\!\!x_1}\qw & \rstick{\!\!\!\!x_1} \\
\lstick{x_2} & \gate{T} & \qw & 
\ctrl{1} & \push{\rule{0em}{1em}}\qw & 
\qw & \ctrl{-1} & 
\ctrl{1} & \qw & \push{\rule{0em}{1em}}\qw & \push{\rule{0em}{1em}}\qw &
\qw & \qw & 
\qw & \push{\rule{0em}{1em}}\qw &\push{\rule{0em}{1em}}\qw & \qw &
\ctrl{-1} & \qw & \rstick{\!\!\!\!x_2} \\
\lstick{x_3} & \qw & \ctrl{-2} & 
\targ & \ustick{\scriptstyle\:\:\:\:\:x_2\oplus x_3}\qw & 
\push{\rule{0em}{1em}}\qw & \gate{T^\dagger} & 
\targ & \ustick{\scriptstyle\!\!\!\!\!\!x_3}\qw & \push{\rule{0em}{1em}}\qw & \push{\rule{0em}{1em}}\qw &
\qw & \ctrl{-2} & \qw & \push{\rule{0em}{1em}}\qw &  \push{\rule{0em}{1em}}\qw &
\gate{T} & \qw & \qw & \rstick{\!\!\!\!x_3}
}
$$
\end{example}

Amy, Maslov and Mosca \cite{amm14} showed that a linear phase operator $U_\poly$ can be synthesized over $\{\cnot, T\}$ given an implementation $\a\in\Z_8^{2^n-1}$ of $\poly$ in time polynomial in the number of non-zero entries of $\a$ -- moreover, this number is linear in the size of the circuit. Their procedure applies each (non-trivial) phase shift $e^{i\frac\pi4 a_\y(\xor{y}{x}{n})}$ by first computing the linear sum $\xor{y}{x}{n}$, then applying $T^{a_\y}$ and uncomputing $\xor{y}{x}{n}$. Recall that $$T^k:=Z^{k_2}P^{k_1}T^{k_0}$$ where $k_2k_1k_0$ is the binary expansion of $k$. Since each $\xor{y}{x}{n}$ is a linear function of the basis state $x_1x_2\cdots x_n$, each value $\xor{y}{x}{n}$ may be computed solely with $\cnot$ gates, giving a total $T$-count equal to the number of odd elements of $\a$ -- we call this the $T$-count of an implementation. While in this work we are only concerned with the $T$-count of the synthesized circuit, $T$-depth can be minimized while keeping $T$-count the same by parallelizing this process through matroid partitioning \cite{amm14}.

The authors used this synthesis algorithm to optimize $T$-count in $\{\cnot, T\}$ circuits by first computing a set of coefficients for the associated linear phase operator $U_\poly$ from the circuit in polynomial time, then synthesizing an equivalent circuit. The remaining linear permutation is also computed and synthesized separately in polynomial time. This procedure has the crucial property that the element $\a$ of $\Z_8^{2^n-1}$ computed has $T$-count at most the $T$-count of the original circuit -- often much lower due to coefficients in the phase polynomial adding and reducing modulo $8$ -- hence the resulting circuit has equal or lesser $T$-count. In particular, we have the following proposition, which relates the $T$-count of a $\{\cnot, T\}$ circuit to the $T$-count of an implementation of the associated phase operator.

\begin{proposition}\label{prop:imp}
Let $U_\poly$ be a linear phase operator in $\P_8(n)$. There exists a circuit over $\{\cnot, T\}$ implementing $U_\poly$ with $T$-count $k$ if and only if there exists $\a\in\Z_8^{2^n-1}$ such that $\poly(\x)= \poly_\a(\x)$ for all $\x\in\F^n$, and $\a$ has at most $k$ odd entries.
\end{proposition}

\section{Decoding-based $T$-count optimization}\label{sec:topt}

While effective at reducing $T$-count, it was noted that the procedure in \cite{amm14} does not always find the minimal $T$-count, as the phase polynomial $\poly$ in question may have many different representations as a weighted sum of linear Boolean functions. For instance, $$4\cdot x_1 + 4\cdot x_2+4\cdot(x_1\oplus x_2)=0\mod 8$$ for all values of $x_1,x_2\in\F$, so $\poly(x_1,x_2)=4\cdot x_1 + 4\cdot x_2+4\cdot(x_1\oplus x_2)$ is an alternative presentation of the zero-everywhere ($\pi/4$) phase polynomial. This implies that further $T$-count optimization may be possible by first finding an implementation of the target phase polynomial with a minimal number of odd coefficients, then synthesizing a circuit. By Proposition~\ref{prop:imp}, this in fact gives a minimal $T$-count circuit. In this section we reduce this problem to a minimum-distance decoding problem and give a $T$-count optimization algorithm based on this decoding.

Given an element $\a$ of $\Z_8^{2^n-1}$, let $[\a]$ be the equivalence class of implementations of $\poly_\a$ -- i.e., $\b\in[\a]$ if and only if $\poly_\a(\x)=\poly_\b(\x)$ for all $\x\in\F^n$ (hence $U_{\poly_\a}=U_{\poly_\b}$). We define $\C$ to be the subset of $\Z_8^{2^n-1}$ giving the zero-everywhere phase polynomial. Note that for any $\a,\b\in\Z_8^{2^n-1}$ and $\x\in\F^n$, $$\poly_\a(\x)+\poly_\b(\x)=\poly_{\a+\b}(\x),$$ so $\C$ is in fact a subgroup of $\Z_8^{2^n-1}$ and moreover, $[\a]=\a+\C$. As a result we see that the problem of finding an implementation of $\poly_\a$ minimizing $T$ count is equivalent to finding an element $\c\in\C$ minimizing the number of odd entries in $\a+\c$. 

In order to find such elements of the coset $\a+\C$, we first need a characterization of the subgroup $\C$. The following lemma gives an explicit set of generators for $\C$ as scaled monomial evaluation vectors with \emph{effective degree} at most $n-4$, giving a type of generalized Reed-Muller code. As the proof is quite technical we give it in Appendix~\ref{sec:proofmain}

\begin{lemma}\label{lem:genset}
$\C$ is generated by $\{2^i\mono{\ev{\fvar}}{y}{n} \mid \y\in\F^n, i\in\Z_3, \wt{\y} - i \leq n-4\}$.
\end{lemma}

With the above set of generators optimization can be performed directly over $\C$, though the particular metric of $T$-count optimality makes such optimization unnatural. As the number of odd entries in an element of $\Z_8^{2^n-1}$ doesn't form a norm, there does not appear to be a natural reduction to lattice problems. Likewise, the number of odd coefficients doesn't make a natural distance metric for (ring) linear codes.

We can instead reduce the optimization problem to a decoding problem over a \emph{binary} code where minimum-distance decoding corresponds exactly to $T$-count optimization. Defining $\res_2:\Z\rightarrow\F$ as the function taking the binary residue of an integer and extending this component-wise to tuples, we see that the number of odd entries in $\a\in\Z_8^{2^n-1}$ is equal to the weight of the binary residue vector, i.e. $\wt{\res_2(\a)}.$ We can further see that $$\wt{\res_2(\a+\c)}=\wt{\res_2(\a)\oplus\res_2(\c)}=\dist{\res_2(\a)}{\res_2(\c)},$$ that, is the $T$-count of $\poly_{\a+\c}$ is the Hamming distance from $\res_2(\a)$ to $\res_2(\c)$. 

Hence, optimizing the number of odd entries in $\a + \c$ over all $\c\in\C$ is exactly the problem of minimum distance decoding $\res_2(\a)$ over $\res_2(\C)$, the set of binary residue vectors of $\C$. We further note that $\res_2(\C)$ is a binary linear code, since $\res_2(\a)\oplus\res_2(\b)=\res_2(\a+\b)\in\res_2(\C)$ for any $\a,\b\in\C$, and as a direct consequence of Lemma~\ref{lem:genset} this code is exactly the ($n-4$)th order, length $2^n-1$ punctured Reed-Muller code.

\begin{theorem}\label{thm:main}
$\res_2(\C)=\R(n-4, n)^*$
\end{theorem}


The remainder of this section discusses some consequences of Theorem~\ref{thm:main}.

\subsection{Upper bounds}

As a consequence of Proposition~\ref{prop:imp} and Theorem~\ref{thm:main}, the covering radius of $\res_2(\C)=\R(n-4,n)^*$ gives a tight upper bound on the number of $T$ gates required to implement a linear phase operator over $\{\cnot, T\}$. Here we mean tight in the sense that there exists a linear phase operator which requires a minimum of $\rho(\R(n-4, n)^*)$ $T$ gates to implement over $\{\cnot, T\}$. While to the best of the authors' knowledge no analytic formula has been found for the covering radius of higher-order Reed-Muller codes, some asymptotic upper bounds are known. In particular, Cohen and Litsyn \cite{cl92} showed that for large $n$ and orders $r$ where $n-r\geq 3$,
$$\rho(\R(r, n))\leq \frac{n^{n-r-2}}{(n-r-2)!}.$$

Since the covering radius of $\R(r, n)^*$ is trivially bounded above by $\rho(\R(r, n))$, we see that for sufficiently large $n$, $\rho(\R(n-4, n)^*)\leq \frac{n^2}{2}-1$. As a result we obtain a new asymptotic bound on the number of $T$ gates required to implement a circuit over $\{\cnot, T\}$.

\begin{theorem}
Any linear phase operator $U_p\in\P_8(n)$ can be implemented with $O(n^2)$ $T$-gates.
\end{theorem}

\subsection{$T$-count optimization}

While the minimal $T$-count of a given phase polynomial $\poly_\a$ can be obtained by finding a minimum distance decoding of $\res_2(\a)$ in $\R(n-4, n)^*$, the decoding itself isn't enough to synthesize a minimal $T$-count circuit. In particular, decoding the binary residue $\res_2(\a)$ of a target tuple $\a\in\Z_8^{2^n-1}$ over $\R(n-4, n)^*$ produces a minimal \emph{residue} $\w=\res_2(\c)$ of a codeword $\c$ in $\C$. To actually produce a minimal $T$-count implementation we need to compute $\c\in\C$ from $\res_2(\c)$ and then synthesize $\a+\c$. Fortunately, there is an easy way to do this, given the following lemma.

\begin{lemma}\label{lem:codebasis}
For all $\y\in\F^n$ with $\wt{\y}\leq n-4$ we have $\mono{\ev{\fvar}}{y}{n}\in\C$.
\end{lemma}

Using Lemma~\ref{lem:codebasis} which follows directly from Lemma~\ref{lem:genset} and the fact that 
set of all monomial evaluation vectors with degree at most $n-4$, $$\B=\{\mono{\ev{\fvar}}{y}{n}\mid\y\in\F^n, \wt{\y}\leq n-4\},$$ forms a basis for $\R(n-4, n)^*$, we can write a decoded word $\w$ over this basis then \emph{reinterpret} the sum over $\Z_8$. Specifically, if
$\w=\b_1\oplus \b_2\oplus\cdots\oplus \b_k$ for some $\b_1,\b_2,\dots, \b_k\in\mathcal{B}$, then we define $\c=\b_1+ \b_2+\cdots+ \b_k$, which by Lemma~\ref{lem:codebasis} is in $\C$, and further note that $$\res_2(\c)=\res_2(\b_1+ \b_2+\cdots+ \b_k)=\res_2(\b_1)\oplus\res_2(\b_2)\oplus\cdots\oplus\res_2(\b_k)=\w.$$ Using this fact we develop an algorithm for the optimization of $T$-count based on Reed-Muller decoding.

\begin{algorithm}[H]
\caption{$T$-optimize($C$)}
\label{alg:cnott}
\begin{algorithmic}[1]
	\STATE Compute phase polynomial coefficients $\a\in\Z_8^{2^n-1}$ from $C$
	\STATE $\w\gets$RM-DECODE($n-4$, $n$, $\res_2(\a)$)
	\STATE Write $\w$ over basis $\mathcal{B}$: $\w=\b_1\oplus \b_2\oplus\cdots\oplus \b_k$
	\STATE $\c\gets \b_1+ \b_2+\cdots+ \b_k$ (mod $8$)
	\STATE SYNTHESIZE($\a+\c$)
\end{algorithmic}
\end{algorithm}

Algorithm~\ref{alg:cnott} summarizes our algorithm for $T$-count optimization in $\{\cnot, T\}$ circuits. For simplicity the algorithm assumes the input circuit implements a linear phase operator -- for more general $\{\cnot, T\}$ circuits the extra linear permutation is synthesized and appended to the end. The algorithm works by computing a set of coefficients implementing the linear phase operator $U_\poly$ computed by the circuit. The vector of residues modulo $2$ is then decoded as $\w$ in $\R(n-4, n)^*$ using the procedure RM-DECODE($n-4$, $n$, $\res_2(\a)$). A vector $\c\in\C$ with binary residue equal to $\w$ is then computed and added to the original set of coefficients, and a circuit is synthesized for the new implementation of $\poly$. In particular, the procedure SYNTHESIZE takes a set of coefficients $\a$ and synthesizes a circuit over $\{\cnot, T\}$ implementing $U_{\poly_\a}$.

The $T$-optimize algorithm is parametric in both the decoder and the synthesis procedure, meaning any variable order Reed-Muller decoder may be used to implement RM-DECODE. If a minimum distance decoder is used, Algorithm~\ref{alg:cnott} synthesizes a minimal $T$-count circuit. Likewise, any synthesis procedure may be used to implement SYNTHESIZE -- for instance, the $T$-depth minimizing $T$-par algorithm \cite{amm14} can be used.

\begin{example}\label{ex:opt}
Consider the circuit below:

\vspace{0.5em}
\centerline{
\Qcircuit @C=0.3em @R=.9em {
	\lstick{x_1} & \ctrl{2}  & \qw          & \ctrl{3}  & \qw         & \rstick{x_1} \qw \\
	\lstick{x_2} & \ctrl{1}  & \gate{Z} & \ctrl{2}   & \qw         & \rstick{x_2} \qw \\
	\lstick{x_3} & \ctrl{-2} & \qw          & \qw         & \qw         & \rstick{x_3} \qw \\
	\lstick{x_4} & \qw         & \qw          & \ctrl{-3} & \gate{S} & \rstick{x_4} \qw
}\qquad
\raisebox{-2.8em}{\!$=$\!}\qquad
\Qcircuit @C=0.3em @R=.3em {
	\lstick{x_1} & \gate{T} & \targ & \qw & \ctrl{2} & \gate{T^\dagger} & \ctrl{1} & \qw & \qw & \ctrl{2} & \targ 
		& \qw          
		& \gate{T} & \targ & \qw & \ctrl{3} & \gate{T^\dagger} & \ctrl{1} & \qw & \qw & \ctrl{3} & \targ 
		& \qw         & \rstick{x_1} \qw \\
	\lstick{x_2} & \gate{T} & \ctrl{-1} & \targ & \qw & \gate{T^\dagger} & \targ & \gate{T^\dagger} & \targ & \qw & \ctrl{-1}
		& \gate{Z} 
		& \gate{T} & \ctrl{-1} & \targ & \qw & \gate{T^\dagger} & \targ & \gate{T^\dagger} & \targ & \qw & \ctrl{-1}  
		& \qw         & \rstick{x_2} \qw \\
	\lstick{x_3} & \gate{T} & \qw & \ctrl{-1} & \targ & \gate{T} & \qw & \qw & \ctrl{-1} & \targ & \qw
		& \qw          
		& \qw & \qw & \qw & \qw & \qw & \qw & \qw & \qw & \qw & \qw         
		& \qw         & \rstick{x_3} \qw \\
	\lstick{x_4} & \qw & \qw & \qw & \qw & \qw & \qw & \qw & \qw & \qw & \qw         
		& \qw          
		& \gate{T} & \qw & \ctrl{-2} & \targ & \gate{T} & \qw & \qw & \ctrl{-2} & \targ & \qw
		& \gate{S} & \rstick{x_4} \qw
}
}\vspace{1em}
By iterating through the circuit and updating the qubit states (see, e.g., \cite{amm14}), we compute the phase polynomial for this operator as
\begin{align*}
	\poly(\x) = 2x_1 &+ 6x_2 + 6(x_1\oplus x_2) + x_3 + 7(x_1\oplus x_3) + 7(x_2\oplus x_3) + (x_1\oplus x_2\oplus x_3) \\
  		&+ 3x_4 + 7(x_1\oplus x_4) + 7(x_2\oplus x_4) + (x_1\oplus x_2\oplus x_4).
\end{align*}
Writing the coefficients of $\poly$ as a $2^n-1$-tuple $\a$ we get $$\a=(2, 6, 6, 1, 7, 7, 1, 3, 7, 7, 1, 0, 0, 0, 0),$$ which has a canonical $T$-count of $8$ -- a reduction of $6$ $T$ gates.

Now we optimize the implementation of $\poly$ further by decoding $$\res_2(\a)=(0, 0, 0, 1, 1, 1, 1, 1, 1, 1, 1, 0, 0, 0, 0)$$ in the code $\R(0, 4)^*$. As $\R(0, 4)^*$ is the set of evaluation vectors for degree 0 binary polynomials, there are exactly two vectors to choose from, corresponding to the zero (zero-everywhere) and constant (one-everywhere) functions. Since the all $1$ vector achieves the minimum distance of $7$ from $\res_2(\a)$, we choose $\w$ to be the all $1$ vector. By Lemma~\ref{lem:codebasis}, $\w=\1$ (mod 2) is already in the space of zero-everywhere polynomials $\C$, so steps 3 \& 4 are trivial and we set $\c=1$ (mod 8). Finally we synthesize a circuit for the tuple $\a+\c = (3, 7, 7, 2, 0, 0, 2, 4, 0, 0, 2, 1, 1, 1, 1),$ corresponding to the phase polynomial 
\begin{align*}
	\poly'(\x) = 3x_1 &+ 7x_2 + 7(x_1\oplus x_2) + 2x_3 + 2(x_1\oplus x_2\oplus x_3) + 4x_4 + 2(x_1\oplus x_2\oplus x_4) \\
  		& +  (x_3\oplus x_4) + (x_1\oplus x_3\oplus x_4) + (x_2\oplus x_3\oplus x_4) + (x_1\oplus x_2\oplus x_3 \oplus x_4).
\end{align*}
A possible circuit implementing $\poly'$ is shown below:

\vspace{0.5em}
\centerline{
\Qcircuit @C=0.5em @R=.3em {
	\lstick{x_1} & \gate{T} & \gate{S}        & \ctrl{1} & \qw       & \targ & \gate{T}                & \targ       & \qw       & \gate{T}
		& \targ & \qw & \qw & \qw & \qw & \qw & \ctrl{1} &  \rstick{x_1} \qw \\
	\lstick{x_2} & \gate{T^\dagger} & \qw & \targ      & \ctrl{1} & \qw   & \gate{T^\dagger} & \ctrl{-1} & \ctrl{2} & \qw
		& \qw & \qw & \ctrl{2} & \qw & \ctrl{2} & \ctrl{1} & \targ &  \rstick{x_2} \qw \\
	\lstick{x_3} & \qw & \gate{S}                & \ctrl{1} & \targ      & \qw   & \gate{S}                & \qw        & \qw        & \qw
		& \qw & \ctrl{1} & \qw & \qw & \qw & \targ & \qw &  \rstick{x_3} \qw \\
	\lstick{x_4} & \qw & \gate{Z}                & \targ      & \qw       & \ctrl{-3} & \gate{T}          & \qw         & \targ     & \gate{T} 
		& \ctrl{-3} & \targ & \targ & \gate{S} & \targ & \qw & \qw &  \rstick{x_4} \qw
}
}\vspace{1em}

Note that this decoding reduces the $T$-count from $14$ (or $8$, as $T$-par type optimizations would obtain) to $7$. Moreover, the number of $T$ gates is equal to the distance from $\res_2(\a)$ to the decoded word $\w$.

\end{example}

It is interesting to note that the minimal $T$-depth of $U_{\poly'}$ above without additional ancillas is $3$, while the minimal ancilla-free $T$-depth of $U_\poly$ is $2$, even though the number of $T$ gates is reduced. Clearly Algorithm~\ref{alg:cnott}, when combined with a $T$-depth optimal synthesis method such as matroid partitioning, does not necessarily obtain the minimal $T$-depth for a given circuit. It remains an open question to determine an efficient method of optimizing $T$-depth over all implementations of a linear phase operator.

\subsection{Complexity}

It may be noted that Algorithm~\ref{alg:cnott} gives a polynomial-time (in $2^n$) reduction from $T$-count optimization over $\{\cnot, T\}$ to minimum-distance decoding in $\R(n-4, n)^*$. We may likewise reduce the minimum-distance decoding problem for $\R(n-4, n)^*$ to $T$-count optimization: given a binary vector $\w\in\F^{2^n-1}$, synthesize $U_{\poly_{\w}}$ over $\{\cnot, T\}$ then optimize the circuit and compute the coefficients $\a\in\Z_8^{2^n-1}$ for the optimized circuit. As a consequence of Theorem~\ref{thm:main}, the vector $\w\oplus\res_2(\a)$ is a minimum distance decoding of $\w$. Assuming the optimized circuit does not have exponentially more gates than a canonical circuit,\footnote{The canonical circuit for any linear phase operator uses $O(n2^{n-1})$ gates.} this reduction is also polynomial in the word length $2^n$, so we see that the problems are in fact polynomial-time equivalent.

This equivalence lends evidence to the difficulty of $T$-count optimization, even in the restricted setting of circuits over $\cnot$ and $T$ gates. In particular, any sub-exponential algorithm for exact optimization of $T$-count over $n$-qubit $\{\cnot, T\}$ circuits induces a polynomial-time minimum-distance decoding algorithm for $\R(n-4, n)$. This can be further reduced to a \emph{linear-time} algorithm by noting that the unitary $U_{\poly_\w}$ above can be implemented with $O(2^n)$ gates using one ancilla and the Gray code to cycle through each of the $2^n$ binary sums of $n$ variables with one $\cnot$ gate each.

In either case it appears very unlikely that an efficient algorithm for minimum-distance decoding the order $n-4$ punctured Reed-Muller code exists. No minimum distance decoding algorithms in time polynomial in $2^n$ \emph{or} the Hamming weight of the received word are currently known for arbitrary order length $2^n$ binary Reed-Muller codes. While some particular orders of Reed-Muller codes have efficient decoders, e.g., order 1, it was shown in \cite{sl83} that minimum-distance decoding for $\R(n-4, n)^*$  is equivalent to the problem of finding a minimal decomposition of a symmetric 3-tensor into symmetric tryads (rank 1 3-tensors), a known hard problem \cite{sl83}.

\section{Rotations of other orders}\label{sec:generalization}

Having shown that minimizing the number of $T$ gates in $\{\cnot, T\}$ circuits is equivalent to minimum distance decoding in $\R(n-4, n)^*$, we now turn our attention to circuits with $Z$-basis rotations of other angles. Specifically, we define the gate $R_z(2\pi/d)$ for any non-zero integer $d$ by $$R_z(2\pi/d):\ket{x}\mapsto e^{\frac{2i\pi}{d}x}\ket{x}.$$ Such gates arise, e.g., in Shor's algorithm \cite{s94} and the Clifford hierarchy \cite{gc99}. Moreover, researchers have recently developed state distillation techniques for these gates, allowing them to be performed fault tolerantly without approximating them over another gate set \cite{lc13, dp15}. Here we develop methods for the optimization of circuits over $\cnot$ and $\R_z(2\pi/d)$ gates to make use of this higher-level structure of many quantum circuits, whether the rotations are then to be approximated over another gate set or implemented directly. 

We define $\P_{d}(n)$ to be the set of $n$-qubit $2\pi/d$ linear phase operators -- that is, $n$-qubit diagonal unitary matrices implementable over $\{\cnot, R_z(2\pi/d)\}$. As in the case of $\pi/4$ linear phase operators, such an operator applies to each basis vector a phase rotation that is a $d$-th root of unity determined by a linear combination of linear functions of its bits. In particular, for any $U\in\P_{d}(n)$, $U$ has the following effect: $$U:\ket{\x}\mapsto e^{\frac{2i\pi}{d}\poly(\x)}\ket{\x},\qquad \poly(\x)=\sum_{\y\in\F^n\setminus\{0\}}a_\y(\xor{y}{x}{n})$$ for some coefficients $\a\in\Z_{d}^{2^n-1}$. As before we call the tuple $\a$ an implementation of $\poly$ and we denote the set of zero-everywhere phase polynomial implementations $\C^d$, defined below as $$\C^d=\{\c\in\Z_{d}^{2^n-1}|\forall\x\in\F^n, \poly_\c(\x) = 0\mod d\}.$$

We first consider the case when $d$ is even, which is a natural generalization of Theorem~\ref{thm:main}. We then examine the case when $d$ is an odd prime power, and finally combine the two results to get a characterization of $\P_d(n)$ for any non-zero integer $d$.

\subsection{Rotations of even order}

Recall that Theorem~\ref{thm:main} was proven by giving a set of generators for $\C=\C^{2^3}$. We can use the same methods to assign a set of generators to $\C^{2^k}$, and likewise derive a generalization of Theorem~\ref{thm:main}. In particular, it turns out that $\C^{2^k}$ is generated by the set of scaled monomial vectors of effective degree at most $n-k-1$.

\begin{lemma}\label{lem:general}
$\C^{2^k}$ is generated by $\{2^i\mono{\ev{\fvar}}{\y}{n} \mid \y\in\F^n, i\in\Z_3, \wt{\y} - i \leq n-k-1\}$.
\end{lemma}

Again the proof of Lemma~\ref{lem:general} is left for Appendix~\ref{sec:proofmain}. Further, as in the $\pi/4$ case, Lemma~\ref{lem:general} implies the following theorem stating that the binary residues of $\C^{2^k}$ are exactly the codewords of the order $n-k-1$ punctured Reed-Muller code of length $2^n-1$.

\begin{theorem}\label{thm:main2}
$\R(n-k-1, n)^*=\res_2(\C^{2^k})$\footnote{Note that using Definition~\ref{def:rm}, the Reed-Muller code $\R(r, n)$ is well defined for $r<0$. In particular, the code is the trivial code $\{\0\}$, corresponding to the fact that no non-trivial zero phase polynomials exist mod $2^k$ when $k<n-1$.}
\end{theorem}

As a consequence of Theorem~\ref{thm:main2}, Algorithm~\ref{alg:cnott} can be adapted to optimize the number of $R_z(2\pi/2^k)$ gates in a linear phase circuit. Recall that the canonical circuit for an implementation of a $\pi/4$ linear phase operator was defined by computing $\xor{y}{x}{n}$ for each nonzero $a_\y$, then applying a sequence of $T$, $P$ and $Z$ gates to achieve the correct power of $e^{i\frac\pi4}$. We may define the canonical circuit for an implementation of any $2\pi/2^k$ linear phase operator in the same way: compute $\xor{y}{x}{n}$ then apply $R_z(2\pi/2^k)^l$ to achieve the correct power of $e^{i\frac\pi{2^k}}$. Under the assumption that $R(2\pi/2^k)$ gates are more expensive than $R(2\pi/2^{k'})$ gates whenever $k> k'$, we define $$R_z(2\pi/2^k)^l:=R_z(2\pi)^{l_k}\cdots R_z(2\pi/2^{l-1})^{l_1}R_z(2\pi/2^k)^{l_0}$$ where $l_k\cdots l_1l_0$ is the binary expansion of $l$. Denoting by $\a>\!\!\!>\!\!\!>i$ the component-wise quotient of $\a$ divided by $2^i$, we find that the number of $R_z(2\pi/2^l)$ gates in the canonical circuit is $\wt{\res_2(\a>\!\!\!>\!\!\!>(k-l))}$ -- the number of components $a_\y$ that have a $1$ in the ($k-l$)th digit of their binary expansion.

The number of rotation gates of any angle $2\pi/2^l$ for $l\leq k$ may then be reduced by decoding $\res_2(\a>\!\!\!>\!\!\!> (k-l))$ in the code $\res_2(\C^l)=\R(n-l-1, n)$ and adding the decoded tuple back into $\a$ (multiplied by the appropriate power of $2$). Such procedures may be a valuable tool for quantum circuits utilizing progressively finer grain $Z$ rotations, such as Shor's algorithm \cite{s94}, either to be later approximated by Clifford+$T$ gates or to be performed directly using state distillation. One potential issue with this method is reducing the number of $R_z(2\pi/2^l)$ may increase the number of $R_z(2\pi/2^{l'})$ gates for any $l'<l$, as seen in Example~\ref{ex:opt}. In most cases smaller angles of rotation are more costly so this is a reasonable trade off, but we leave it as an open question to find a general algorithm for optimizing the total cost of all rotation gates in a $\{\cnot, R_z(2\pi/2^k)\}$ circuit.

\subsection{Rotations of odd order}
A natural question is whether rotation gates of other prime power orders admit similar relationships to known codes. To the contrary, we show that for any odd prime $p$ and integer $k$, there are no non-trivial phase polynomials that are zero-everywhere mod $p^k$.

%

\begin{lemma}\label{lem:primepow}
For all odd primes $p$ and non-negative integers $k$, given any non-zero tuple $\a\in\Z_{p^k}^{2^n-1}$, there exists $\x\in\F^n$ such that $$\poly_\a(\x) \neq 0\mod p^k.$$
\end{lemma}

To prove Lemma~\ref{lem:primepow}, we first introduce the \emph{multilinear representation} of a phase polynomial. In particular, given a tuple $\a\in\Z_{p^k}^{2^n-1}$ the multilinear polynomial function defined by $\a$ is given by $$Q_\a(\x)= \sum_{\y\in\FF{n}} a_\y\cdot \mono{x}{y}{n}.$$ The result then follows from two facts:
\begin{enumerate}
	\item there are no non-trivial zero-everywhere multilinear polynomials modulo $\Z_{p^k}$, and
	\item for every multilinear polynomial over $\Z_{p^k}$, there exists a unique equivalent phase polynomial over $\Z_{p^k}$.
\end{enumerate}

The first fact follows from the observation that the set of all non-constant monomial evaluation vectors $$\{\mono{\ev{\fvar}}{y}{n} | \y\in\FF{n}\}$$ is linearly independent over any integer ring. In particular, for any $\y\in\FF{n}$, the vector $\mono{\ev{\fvar}}{y}{n}$ contains a leading $1$ at the $\y$th index (see e.g., Table~\ref{tab:evals}), and hence the set of all such vectors is trivially linearly independent in any integer ring.
\begin{proposition}\label{prop:multzero}
For all odd primes $p$ and non-negative integers $k$, given any non-zero tuple $a\in\Z_{p^k}^{2^n-1}$, there exists $\x\in\F^n$ such that $$Q_\a(\x) \neq 0\mod p^k.$$
\end{proposition}

For the second fact, recall the modular identity $$2xy = x + y - (x\oplus y) \mod p^k$$ for any $x, y\in\F$. Since $2$ is coprime with $p^k$ it has a multiplicative inverse in $\Z_{p^k}$, hence we can rewrite this identity as 
$$
	xy = 2^{-1}x + 2^{-1}y -2^{-1}(x\oplus y) \mod p^k.
$$
The equation above can be used to rewrite a monomial $x_{i_1}x_{i_2}\cdots x_{i_m}$ in the form of a phase polynomial:
\begin{align*}
	(x_{i_1}x_{i_2})\cdots x_{i_m} &= (2^{-1}x_{i_1} + 2^{-1}x_{i_2} - 2^{-1}(x_{i_1} \oplus x_{i_2}))\cdots x_{i_n} \mod p^k \\
		&= 2^{-1}x_{i_1}\cdots x_{i_n} + 2^{-1}x_{i_2}\cdots x_{i_n} - 2^{-1}(x_{i_1} \oplus x_{i_2})\cdots x_n \mod p^k,
\end{align*}
where each term in the second line has degree $m-1$ and hence the monomial can be recursively reduced to the form of a phase polynomial. Uniqueness further follows from Proposition~\ref{prop:multzero}, as if two distinct multilinear polynomials $Q_{\a}$ and $Q_{\b}$ reduced to the same phase polynomial, we would have $Q_{\a-\b}(\x)=0\mod p^k$ for all $\x\in\F^n$ but $\a+\b\neq \0$, a contradiction.

\begin{proposition}\label{prop:endo}
For any odd prime $p$ and positive integer $k$, given a tuple $\a\in\Z_{p^k}^{2^n-1}$ there exists some unique $\b\in\Z_{p^k}^{2^n-1}$ such that for all $\x\in\Z_2^n$,
$$
Q_{\a}(\x) = \poly_{\b}(\x) \mod p^k.
$$
\end{proposition}

Note that Proposition~\ref{prop:endo} \emph{does not} hold for even prime powers $p^k$, as it requires $p^k$ to be coprime with $2$ in order to rewrite a monomial as a weighted sum of parities.

Propositions~\ref{prop:multzero} and \ref{prop:endo} together imply that there exists an isomorphism between multilinear and phase polynomial representations of pseudo-Boolean functions modulo powers of odd primes, and moreover that there are no non-trivial zero-everywhere multilinear polynomials and hence phase polynomials. We formalize this intuition below.

\begin{proof}[Proof of Lemma~\ref{lem:primepow}]
Suppose $\a\in\Z_{p^k}$ is non-zero for some odd prime $p$ and non-negative integer $k$. 
By Proposition~\ref{prop:endo} and the fact that there are the same number of multilinear and phase polynomials over $\Z_{p^k}$, there exists a unique tuple $\b\in\Z_{p^k}^{2^n-1}$ such that $\poly_{\a}(\x)=Q_{\b}(\x)$ for all $\x\in\F^n$. Now by Proposition~\ref{prop:multzero}, there exists $\x\in\F^n$ such that $$\poly_\a(\x) = Q_{\b}(\x) \neq 0\mod p^k$$ as required.
\end{proof}

From the perspective of optimizing phase gates, Lemma~\ref{lem:primepow} asserts that each element of $\Z_{p^k}^{2^n-1}$ corresponds to a unique $n$-qubit unitary implementable over $\CNOT$ and $R(2\pi/p^k)$. Given such a circuit, an implementation minimizing the number of the minimal number of $R(2\pi/p^k)$ may then be obtained by first computing the corresponding element of $\Z_{p^k}^{2^n-1}$ and resynthesizing, which can be performed in polynomial time.

\subsection{Rotations of arbitrary order}

It is worth noting that Lemma~\ref{lem:primepow} above is in a sense the complement to Lemma~\ref{lem:general}. Together, they give a characterization of the linear phase operators with rotation gates that form arbitrary cyclic groups. In particular, the set of phase polynomials which are zero-everywhere mod $d$ for any non-zero $d\in\Z$ is given by scaling the zero-everywhere polynomials for the even part of $d$. 

\begin{theorem}
Let $d$ be any non-negative integer, and suppose the prime factorization of $d$ is $2^{d_1}3^{d_2}5^{d_3}\cdots$. Then $$\C^d=\C^{2^{d_1}}\cdot 3^{d_2}5^{d_3}\cdots.$$
\end{theorem}
\begin{proof}
Th inclusion of $\C^{2^{d_1}}\cdot 3^{d_2}5^{d_3}\cdots$ in $\C^d$ is trivial, so let $\a$ be some tuple in $\Z_d^{2^n-1}$ and suppose $P_{\a}(\x) = 0\mod d$ for all $\x\in\F^n$. 

Clearly $P_{\a}(\x) = 0\mod d$ for all $\x\in\F^n$ if and only if $P_{\a}(\x) = 0\mod p_i^{d_i}$ for all $\x\in\F^n$ and prime power $p_i^{d_i}$ in the prime factor decomposition of $d$. However, by Lemma~\ref{lem:primepow} for any $p\neq 2$, $P_{\a}(\x) = 0\mod p_i^{d_i}$ if and only if $\a=0\mod p_i^{d_i}$, so $\a=\a'\cdot p_i^{d_i}$ where $\a'\in\Z_{d/p_i^{d_i}}$ and $P_{\a'}(\x) = 0\mod d/p_i^{d_i}.$ Repeating for all odd primes, we see that $$\a=\a'\cdot 3^{d_2}5^{d_3}\cdots$$ for some $\a'\in\Z_{2^{d_1}}$ where $\poly_{\a'}(\x) = 0\mod 2^{d_1}$ for all $\x\in\F^n$. Hence $\a'\in\C^{2^{d_1}}$ and so $\a\in\C^{2^{d_1}}\cdot 3^{d_2}5^{d_3}\cdots$ as required.
\end{proof}

\section{Experiments}\label{sec:experiments}

We implemented Algorithm~\ref{alg:cnott} in $T$-par \cite{tpar} as an optimization pass in the resynthesis procedure. $T$-par optimizes circuits over the Clifford+$T$ gate set by computing a representation using exponential sums, then resynthesizing. As our algorithm presently applies to $\cnot$ and phase gates, we break up the input circuit into $\{\cnot, T\}$ subcircuits, each of which is then optimized individually.

We implemented and tested the algorithm with two Reed-Muller decoders -- a majority logic decoder due to Reed~\cite{r54}, and a modern recursive decoder due to Dumer~\cite{d04}. The former has complexity in $O(2^{2n})$ for an $n$-qubit circuit while the latter has a significantly lower complexity of $O(2^n)$. While both of these algorithms are exponential in the number of qubits $n$, we nonetheless obtain reasonable performance for large circuits by storing and operating directly on compressed vector representations. In order to optimize these large circuits we chose relatively fast decoders over minimum-distance decoders.

\subsection{Evaluation}

Algorithm~\ref{alg:cnott} was evaluated on a suite of benchmark quantum circuits, drawn from the literature and the Reversible Logic Benchmarks page \cite{m11}. The majority of circuits tested are reversible circuits, though some specifically quantum circuits were also examined. Toffoli gates were replaced with a Clifford+$T$ implementation using $7$ $T$-gates \cite{ammr13}, and multiple control Toffolis were expanded into two-control Toffoli gates using one zero initialized ancilla (see, e.g., \cite{nc00}).

\begin{table}
\footnotesize
\caption{$T$-count optimization results. $n$ reports the number of qubits in the circuit. $T$-counts are recorded for the original circuit, after optimization by $T$-par, and after optimization by Algorithm~\ref{alg:cnott} with either the majority logic or recursive decoder.}
\label{tab:results}
\centering
\begin{tabular}[t]{lr@{\hspace{1em}}r@{\hspace{1em}}r@{\hspace{1em}}r@{\hspace{1em}}r} \toprule
Benchmark & $n$ & \multicolumn{4}{c}{$T$-count} \\ \cmidrule(l{0pt}r{0pt}){3-6}
& & original & $T$-par & majority & recursive \\
Grover$_5$ \cite{g96}\footnotemark & 9 & 140 & 52 & 52 & 52 \\ 
Mod 5$_{4}$ \cite{m11} & 5 & 28 & 16 & 16 & 16 \\ 
VBE-Adder$_{3}$ \cite{vbe96} & 10 & 70 & 24 & 24 & 24 \\ 
CSLA-MUX$_{3}$ \cite{vi05} & 15 & 70 & 62 & 62 & 58 \\
CSUM-MUX$_{9}$ \cite{vi05} & 30 & 196 & 140 & 84 & 76 \\ 
QCLA-Com$_{7}$ \cite{dkrs06} & 24 & 203 & 95 & 94 & 153 \\
QCLA-Mod$_{7}$ \cite{dkrs06} & 26 & 413 & 249 & 238 & 299 \\
QCLA-Adder$_{10}$ \cite{dkrs06} & 36 & 238 & 162 & -- & 188 \\ 
Adder$_{8}$ \cite{ttk10}& 24 & 399 & 215 & 213 & 249 \\
RC-Adder$_{6}$ \cite{cdkp04}& 14 & 77 & 63 & 47 & 47 \\ 
Mod-Red$_{21}$ \cite{ms12}& 11 & 119 & 73 & 73 & 73 \\
Mod-Mult$_{55}$ \cite{ms12} & 9 & 49 & 37 & 35 & 35 \\ 
Mod-Adder$_{1024}$ \cite{m11} & 28 & 1995 & 1011 & 1011 & 1011 \\ 
BCSD$_2$ \cite{f13} & 9 & 14 & 14 & 2 & 2 \\
BCSD$_4$ \cite{f13} & 14 & 20 & 20 & 4 & 4 \\
BCSD$_8$ \cite{f13} & 21 & 32 & 32 & 8 & 8 \\ 
Cycle $17\_3$ \cite{m11}& 35 & 4739 & 1945 & 1944 & 1982 \\ 
HWB$_6$  \cite{m11} & 7 & 105 & 71 & 75 & 75 \\
HWB$_8$  \cite{m11} & 12 & 5887 & 3551 & 3531 & 3531 \\
$n$th-prime$_6$ \cite{m11} & 9 & 812 & 402 & 400 & 400 \\
$n$th-prime$_8$ \cite{m11} & 12 & 6671 & 4047 & 4045 & 4045 \\ 
$\vdots$ & $\vdots$ & $\vdots$ & $\vdots$ & $\vdots$ & $\vdots$\vspace*{1mm} \\ \bottomrule
\end{tabular} \hfill
\begin{tabular}[t]{lr@{\hspace{1em}}r@{\hspace{1em}}r@{\hspace{1em}}r@{\hspace{1em}}r} \toprule
Benchmark & $n$ & \multicolumn{4}{c}{$T$-count} \\ \cmidrule(l{0pt}r{0pt}){3-6}
& & original & $T$-par & majority & recursive \\
$\vdots$ & $\vdots$ & $\vdots$ & $\vdots$ & $\vdots$ & $\vdots$ \\
GF($2^4$)-Mult \cite{mmcp09} & 12 & 112 & 68 & 68 & 68 \\
GF($2^5$)-Mult \cite{mmcp09} & 15 & 175 & 111 & 111 & 101 \\
GF($2^6$)-Mult \cite{mmcp09} & 18 & 252 & 150 & 150 & 144 \\
GF($2^7$)-Mult \cite{mmcp09} & 21 & 343 & 217 & 217 & 208 \\
GF($2^8$)-Mult \cite{mmcp09} & 24 & 448 & 264 & 264 & 237 \\
GF($2^9$)-Mult \cite{mmcp09} & 27 & 567 & 351 & -- & 301 \\
GF($2^{10}$)-Mult \cite{mmcp09} & 30 & 700 & 410 & -- & 410 \\
GF($2^{16}$)-Mult \cite{mmcp09} & 48 & 1792 & 1040 & -- & --  \\ 
GF($2^{32}$)-Mult \cite{mmcp09} & 96 & 7168 & 4128 & -- & --  \\ 
Hamming$_{15}$ (low)  \cite{m11} & 17 & 161 & 97 & 97 & 97 \\
Hamming$_{15}$ (med)  \cite{m11} & 17 & 574 & 230 & 230 & 230 \\
Hamming$_{15}$ (high)  \cite{m11} & 20 & 2457 & 1019 & 1019 & 1019 \\
QFT$_4$ \cite{nc00} & 5 & 69 & 67 & 67 & 67 \\ 
$\Lambda_3(X)$ -- \cite{bbcdmsssw95} & 5 & 28 & 16 & 16 & 16 \\
\hspace{3.05em} -- \cite{nc00} & 5 & 21 & 15 & 15 & 15 \\
$\Lambda_4(X)$ -- \cite{bbcdmsssw95} & 7 & 56 & 28 & 28 & 28 \\
\hspace{3.05em} -- \cite{nc00} & 7 & 35 & 23 & 23 & 23 \\
$\Lambda_5(X)$ -- \cite{bbcdmsssw95} & 9 & 84 & 40 & 40 & 40 \\
\hspace{3.05em} -- \cite{nc00} & 9 & 49 & 31 & 31 & 31 \\
$\Lambda_{10}(X)$ -- \cite{bbcdmsssw95} & 19 & 224 & 100 & 100 & 100 \\
\hspace{3.05em} -- \cite{nc00} & 19 & 119 & 71 & 71 & 71 \\ \bottomrule
\end{tabular}
\end{table}\footnotetext{Grover's search is performed with 4 iterations using the oracle $f(\x)=\neg x_1\land \neg x_2\land x_3\land x_4\land\neg x_5$.}

Table~\ref{tab:results} reports the $T$-count of circuits optimized with both $T$-par alone, and with Algorithm~\ref{alg:cnott} using either the majority logic or recursive decoder applied to $\{\cnot, T\}$ subcircuits. All experiments were run on with a 2.4GHz quad-core Intel Core i7 processor running Linux and 8GB of RAM. Each benchmark had a timeout of 30 minutes -- instances where the algorithm failed to report a result within the timeout are identified with a dash.

On average, Algorithm~\ref{alg:cnott} reduced $T$-count by 6\% for both the majority logic decoder and the recursive decoder compared to $T$-par. While the recursive decoder produced the best results in some cases, notably the Galois field multipliers, and failed less often, for many benchmarks it reported significantly \emph{increased} $T$-counts compared to $T$-par. Majority logic decoding by comparison typically produced less $T$-reduction, though it consistently resulted in circuits with equal or lesser $T$-count than that reported by $T$-par. Counter-intuitively this appears to result from the recursive decoder actually doing a \emph{better} job optimizing $T$-count -- after the recursive decoder performs significant rewrites on individual $\{\cnot, T\}$ subcircuits, $T$-par has less opportunity to optimize $T$-gates across subcircuit boundaries. A natural direction of future research is to extend decoding-based optimization to $\{H, \cnot, T\}$ circuits in order to make use of the additional $T$-count reductions possible across subcircuit boundaries.

While the $T$-count reductions over $T$-par are minor compared to the initial jump from the original $T$-count, the results clearly demonstrate that further $T$-count optimization beyond the $T$-par algorithm is possible. In the most significant case a $T$-count reduction of 75\% was reported for the benchmark BCSD$_8$ with both decoders, though as the benchmark performs state distillation and relies on certain properties of the circuit for fault tolerance, such an optimization is not likely useful. Note that it may be possible to achieve better $T$-count with other Reed-Muller decoders as well. We leave exploration of effective decoders as an avenue for future work.

As an additional note, while we don't consider $T$-depth optimization in this paper, reductions to $T$-count in some benchmarks allow further reductions to $T$-depth using matroid partitioning. In the extreme case, $T$-depth in CSUM-MUX$_9$ was reduced from 11 to 6 using the recursive decoder, providing strong evidence that reducing $T$-count is an effective means of optimizing $T$-depth.

\section{Conclusion}\label{sec:conclusion}

In this paper we have answered the question previously posed in \cite{amm14} of whether there exist identities which can be used to reduce the $T$-cost of a phase polynomial over $\cnot$ and $T$ gates. We gave a concrete set of generators for the entire set of identities and have shown that, when restricted to $T$-count optimization, these identities correspond exactly to the punctured Reed-Muller code of length $2^n-1$ and order $n-4$. From this correspondence we developed a $T$-count optimization procedure which uses Reed-Muller decoders to reduce the $T$-cost of a phase polynomial and is optimal when a minimum distance decoder is used, as well as gave a new upper bound on the $T$-count of $\{\cnot, T\}$ circuits. We also looked at the question of optimizing phase polynomials corresponding to other $Z$-basis rotation gates, giving a concrete set of generators for the set of identities over rotations of any finite order.

A natural continuation of this programme is to find methods for minimizing the $T$-count of quantum circuits over a universal set of gates -- for instance, the standard Clifford+$T$ set generated by $\{H, \cnot, T\}$. Our methods give both an upper bound of $O(k\cdot n^2)$ $T$-gates for a circuit containing $k$ Hadamard gates, as well as a concrete algorithm which achieves this bound when using a minimum distance decoder. On the other hand, the $(n+k)$-variate phase polynomial for an entire $k$-Hadamard circuit over $\{H, \cnot, T\}$ may be computed and optimized directly \cite{amm14}, giving an upper bound of $O((n+k)^2)$ $T$ gates with the caveat that the resulting operator may not be implementable with only $n$ qubits. In either case the minimal $T$-count depends on the Hadamard cost of the circuit which may itself be reduced, implying that unlike the $\{\cnot, T\}$ case, the minimal $T$-count of a Clifford+$T$ circuit may not be achievable simply by rewriting its phase polynomial. We leave it as a question for future research to determine the relationship between phase polynomials, Hadamard gates and ancillas, as well as upper bounds and methods for finding the exact minimal $T$-count of Clifford+$T$ circuits.

\section*{Acknowledgments}
The authors would like to thank Adam Paetznick and Neil J. Ross for helpful discussions and comments on an earlier version. Matthew Amy is supported in part by the Natural Sciences and Engineering Research
Council of Canada. Michele Mosca is supported by Canada's NSERC, MPrime, CIFAR, and CFI. IQC and Perimeter Institute are supported in part by the Government of Canada and the Province of Ontario.

\bibliographystyle{IEEEtran}
\bibliography{paper}

\appendix

\section{Generators of $\C^{2^k}$}\label{sec:proofmain}

In this section we give an explicit set of generators for the space of zero-everywhere phase polynomials modulo powers of $2$. In particular, we give proofs of Lemma~\ref{lem:genset} and the general version, Lemma~\ref{lem:general}.

\subsection{The monomial basis}

Our proof relies on a connection between the binary evaluations of polynomials over $\Z_8$ and the module $\Z_8^{2^n-1}$. In particular, consider the set of degree at most $n-1$ monomial (Boolean) evaluation vectors $$\{\mono{\ev{\fvar}}{y}{n}\mid\y\in\FO{n}\}.$$ We show that this set of vectors, under the natural inclusion of $\F$ in $\Z_8$, forms a generating set for $\Z_8^{2^n-1}$ -- moreover, since each such vector is linearly independent over $\F^{2^n-1}$ and hence also linearly independent over $\Z_8^{2^n-1}$, this set is in fact a basis. We call this basis the \emph{monomial basis} of $\Z_8^{2^n-1}$.

\begin{lemma}\label{lem:basis}
$\{\mono{\ev{\fvar}}{y}{n}\mid\y\in\FO{n}\}$ is a basis of $\Z_8^{2^n-1}$
\end{lemma}
\begin{proof}

We first note that the set of all non-constant monomial evaluation vectors, $\{\mono{\ev{\fvar}}{y}{n}\mid\y\in\FF{n}\}$, is a basis for the module $\Z_8^{2^n-1}$. In particular, for any $\y\in\FF{n}$ the vector $\mono{\ev{\fvar}}{y}{n}$ contains a leading $1$ at the $\y$th index (e.g., Table~\ref{tab:evals}), and hence any tuple of $\Z_8^{2^n-1}$ may be written as a linear combination over this set. It therefore suffices to prove that $\ev{\fvar}_1\ev{\fvar}_2\cdots\ev{\fvar}_n$ is in the span of $\{\mono{\ev{\fvar}}{y}{n}\mid\y\in\FO{n}\}$.

It may be observed that over $\F$, the set of \emph{all} monomial evaluation vectors is linearly dependent, and in particular that $$\bigoplus_{\y\in\F^n}\mono{\ev{\fvar}}{y}{n}=\0$$ since every input evaluates to $1$ for an even number of monomials. Further, as $\res_2$ is homomorphic we have $$\bigoplus_{\y\in\F^n}\mono{\ev{\fvar}}{y}{n}=\res_2\left(\sum_{\y\in\F^n}\mono{\ev{\fvar}}{y}{n}\right)=\0$$ and so $\sum_{\y\in\F^n}\mono{\ev{\fvar}}{y}{n}=\a$ for some $\a\in\Z_8^{2^n-1}$ such that $\res_2(\a)=0$. If we write $\a$ over the basis $\{\mono{\ev{\fvar}}{y}{n}\mid\y\in\FF{n}\}$ and move all instances of $\ev{\fvar}_1\ev{\fvar}_2\cdots\ev{\fvar}_n$ to the left we see $$b\cdot\ev{\fvar}_1\ev{\fvar}_2\cdots\ev{\fvar}_n=\a' - \sum_{\y\in\FO{n}}\mono{\ev{\fvar}}{y}{n}$$ where $\a'$ is in the span of $\{\mono{\ev{\fvar}}{y}{n}\mid\y\in\F^n\setminus\{\0,\1\}\}$ and $b\in\Z_8$. 

Now suppose $b$ is even. Then
\begin{align*}
(b-1)\cdot \ev{\fvar}_1\ev{\fvar}_2\cdots\ev{\fvar}_n &= \a' - \sum_{\y\in\F^n}\mono{\ev{\fvar}}{y}{n} \\
\res_2(\ev{\fvar}_1\ev{\fvar}_2\cdots\ev{\fvar}_n) &= \res_2(\a') + \res_2(\sum_{\y\in\F^n}\mono{\ev{\fvar}}{y}{n}) \\
\ev{\fvar}_1\ev{\fvar}_2\cdots\ev{\fvar}_n &= \res_2(\a').
\end{align*}
Since $\a'$ is in the span of $\{\mono{\ev{\fvar}}{y}{n}\mid\y\in\F^n\setminus\{\0,\1\}\}$ over $\Z_8$, $\res_2(\a')$ is in its span over $\F$ and hence may be written over this basis. However, the set of all monomial evaluation vectors of degree at least 1 is linearly independent over $\F$, so we arrive at a contradiction. 

Thus $b$ is odd and as such has a multiplicative inverse in $\Z_8$. Hence, $$\ev{\fvar}_1\ev{\fvar}_2\cdots\ev{\fvar}_n=b\cdot b^{-1}\cdot\ev{\fvar}_1\ev{\fvar}_2\cdots\ev{\fvar}_n=b^{-1}\cdot\left(\a' - \sum_{\y\in\FO{n}}\mono{\ev{\fvar}}{y}{n}\right).$$

\end{proof}

Lemma~\ref{lem:basis} tells us that any element $\a$ of $\Z_8^{2^n-1}$ is the vector of evaluations for some pseudo-Boolean polynomial function $f:\F^n\rightarrow\Z_8$ where $f(\fvar_1, \fvar_2,\dots, \fvar_n)=\sum_{\y\in\FO{n}}b_\y \mono{\fvar}{y}{n},$ and hence $\ev{f}=\a=\sum_{\y\in\FO{n}}b_\y \mono{\ev{\fvar}}{y}{n}$ in the monomial basis. Moreover, since $$\res_2(\a) = \sum_{\y\in\FO{n}}\res_2(b_\y) \mono{\ev{\fvar}}{y}{n},$$ $\res_2(\a)$ is the evaluation vector of a Boolean polynomial function with degree at most $\deg(f)$.

\subsection{Evaluating $\poly_\a$}\label{sec:eval}

The next step in our proof is to give an analytic formula for the value of a phase function $\poly_\a$ applied to a vector $\x\in\F^n$ as a function of the degree of the polynomial form of $\a$. Specifically, we show that $\poly_\a(\x)$ is equal to a linear combination of the Hamming weights -- numbers of solutions -- of certain Boolean polynomials arising from the multiplication of a monomial with a degree $1$ polynomial.

Consider the value of a phase polynomial $\poly_\a$ at $\x\in\F^n$: $$\poly_\a(\x)=\sum_{\y\in\FF{n}} a_\y(\xor{y}{x}{n}).$$ We can view the above equation as an inner product, since the value $\xor{y}{x}{n}$ is the $\y$th component of the evaluation vector $\xor{x}{\ev{\fvar}}{n}$. 

Formally, we define $\langle \a, \b\rangle$ for $\a,\b\in\Z_8^{2^n-1}$ as $\sum_{i=1}^{2^n-1}a_ib_i$. Note that $$\langle\a+\b, \c\rangle=\langle\a, \c\rangle+\langle\b,\c\rangle$$ for any $\a,\b,\c\in\Z_8^{2^n-1}$ since the inner product is linear in either argument over $\Z$, and hence also $\Z_8$. Using this observation, we give an explicit formula for $\poly_\a(\x)$ as a function of the basis vectors appearing in $\a$:

\begin{lemma}\label{lem:eval}
Let $\normalfont \a\in\Z_8^{2^n}$ and suppose $\normalfont \a=\sum_{\y\in\FO{n}}b_\y \mono{\ev{\fvar}}{y}{n}$ in the monomial basis. Then $$\normalfont \poly_\a(\x)=\sum_{\y\in\FO{n}}b_\y \wt{(\mono{\ev{\fvar}}{y}{n})(x_1\ev{\fvar}_1\oplus x_2\ev{\fvar}_2\oplus\cdots x_n\ev{\fvar}_n)}.$$
\end{lemma}
\begin{proof}
\begin{align*}
\poly_\a(\x) &= \sum_{\y\in\FF{n}} a_\y(\xor{y}{x}{n}) \\
	&= \sum_{\y\in\FF{n}} a_\y(\xor{x}{\ev{\fvar}}{n})_\y \\
	&= \langle\a,  \xor{x}{\ev{\fvar}}{n}\rangle \\
	&= \sum_{\y\in\FO{n}}b_\y \langle \mono{\ev{\fvar}}{y}{n},  \xor{x}{\ev{\fvar}}{n}\rangle \\
	&= \sum_{\y\in\FO{n}}b_\y \wt{(\mono{\ev{\fvar}}{y}{n})(\xor{x}{\ev{\fvar}}{n})}.
\end{align*}
\end{proof}

The value of $\wt{(\mono{\ev{\fvar}}{y}{n})(\xor{x}{\ev{\fvar}}{n})}$ in Lemma~\ref{lem:eval} above may be restated as the number of solutions to the equation $(\mono{\fvar}{y}{n})(\xor{x}{\fvar}{n})=1$. Fortunately, this number of a simple function of the degree of the polynomial, as the following Lemma shows.

\begin{lemma}\label{lem:weight}
For any $\x, \y\in\F^n$, $$\wt{(\mono{\ev{\fvar}}{y}{n})(\xor{x}{\ev{\fvar}}{n})} = 2^{n-\deg((\mono{\fvar}{y}{n})(\xor{x}{\fvar}{n}))}$$
if $(\mono{\fvar}{y}{n})(\xor{x}{\fvar}{n})\neq 0$, or $0$ otherwise.
\end{lemma}
\begin{proof}

Clearly if $(\mono{\fvar}{y}{n})(\xor{x}{\fvar}{n}) = 0$, then $\wt{(\mono{\ev{\fvar}}{y}{n})(\xor{x}{\ev{\fvar}}{n})} = 0$ as required, so suppose instead that $(\mono{\fvar}{y}{n})(\xor{x}{\fvar}{n})\neq 0$. Since $\xor{x}{\fvar}{n}$ has a degree of $1$, $$\deg((\mono{\fvar}{y}{n})(\xor{x}{\fvar}{n}))=\deg(\mono{\fvar}{y}{n})\text{, or }\deg(\mono{\fvar}{y}{n})+1.$$

Consider the former case. Clearly $\mono{\fvar}{y}{n}$ is a linear combination involving only variables present in $\mono{\fvar}{y}{n}$. Then by the equivalence $\fvar_i^2=\fvar_i$ for any $i$, $$(\mono{\fvar}{y}{n})(\xor{x}{\fvar}{n})=(\wt{\x} \mod 2)\cdot\mono{\fvar}{y}{n}.$$ Since $(\mono{\fvar}{y}{n})(\xor{x}{\fvar}{n})\neq 0$, it must be the case that $\wt{\x}=1\mod 2$. As $\fvar_{i_1}\fvar_{i_2}\cdots \fvar_{i_j}=1$ exactly when $\fvar_{i_1}=\fvar_{i_2}=\cdots = \fvar_{i_j}=1$, we see that $\mono{\fvar}{y}{n}=1$ has $2^{n-\wt{\y}}$ solutions, and hence $$\wt{\mono{\ev{\fvar}}{y}{n}}=2^{n-\deg(\mono{\fvar}{y}{n})}.$$

Now consider the latter case, $\deg((\mono{\fvar}{y}{n})(\xor{x}{\fvar}{n}))=\deg(\mono{\fvar}{y}{n})+1$. We know the linear combination $\xor{x}{\fvar}{n}$ must contain a non-zero multiple of some variable not in the monomial $\mono{\fvar}{y}{n}$. Without loss of generality assume $\mono{\fvar}{y}{n}$ and $\xor{x}{\fvar}{n}$ have no variables in common, as otherwise we may write $$(\mono{\fvar}{y}{n})(\xor{x}{\fvar}{n})=(\mono{\fvar}{y}{n})(c\oplus\xor{x'}{\fvar}{n})$$ for some $c\in\F$ such that $\xor{x'}{\fvar}{n}$ involves none of the variables in $\mono{\fvar}{y}{n}$. Since $c\oplus\xor{x'}{\fvar}{n}=1$ for exactly half of the $2^{n-\wt{\y}}$ valuations where $\mono{\fvar}{y}{n}=1$, we see that there are $2^{n-\wt{\y}-1}$ solutions, hence $$\wt{(\mono{\ev{\fvar}}{y}{n})(\xor{x}{\ev{\fvar}}{n})} = 2^{n-\deg(\mono{\fvar}{y}{n}) - 1}$$ as required.
\end{proof}

In general, it is not the case that $\wt{\ev{f}}=2^{n-\deg(f)}$ for an $n$-variate Boolean polynomial function $f$. In particular, consider $f(\fvar_1, \fvar_2,\dots, \fvar_n)=1\oplus \fvar_1\fvar_2\cdots \fvar_i$. Since $\fvar_1\fvar_2\cdots \fvar_i=1$ has $2^{n-i}$ solutions, $$\wt{\ev{f}}=2^n-2^{n-i}\neq 2^{n-\deg(f)}.$$

\subsection{An explicit set of generators}\label{sec:slice}

From Lemma~\ref{lem:weight} it's immediate that if $\a\in\Z_8^{2^n-1}$ may be written over the monomial basis with degree at most $n-4$, then $\poly_\a(\x) = 0\mod 8$ for any $\x$ and so $\a\in\C$. However, it may be the case that $\a$ contains monomials with degree greater than $n-4$ and yet are still in $\C$. For instance, consider $\a = 2\cdot\ev{\fvar}_1\ev{\fvar}_2\cdots\ev{\fvar}_{n-3}$. For any $\x\in\F^n$, 
\begin{align*}
\poly_\a(\x) 
	&= 2\cdot\wt{(\ev{\fvar}_1\ev{\fvar}_2\cdots\ev{\fvar}_{n-3})(\xor{x}{\ev{\fvar}}{n})} \\
	&= 2\cdot 2^{n-\deg((\fvar_1\fvar_2\cdots \fvar_{n-3})(\xor{x}{\fvar}{n}))} \\
	&=0\mod 8.
\end{align*}
In this case we have $\a\in\C$ and $\res_2(\a) = \0\in\R(n-4, n)^*$ for $n\geq 4$ even though as a polynomial over $\Z_8$, $\a$ has degree greater than $n-4$. On the other hand, in some sense, with regard to Lemma~\ref{lem:weight}, the term $2\cdot \fvar_1\fvar_2\cdots \fvar_{n-3}=2^1\cdot \fvar_1\fvar_2\cdots \fvar_{n-3}$ has an \emph{effective} degree of $\deg(\fvar_1\fvar_2\cdots \fvar_{n-3})-1=n-4$, since $2\cdot\wt{(\ev{\fvar}_1\ev{\fvar}_2\cdots\ev{\fvar}_{n-3})(\xor{x}{\ev{\fvar}}{n})} = 0$, $2^{n-(n - 4)}=2^4$ or $2^{n-(n-3)}=2^3$. 

We define the effective degree of a term of the form $2^i\cdot\mono{\fvar}{y}{n}$ to be $\wt{\y} - i$. Moreover, we let the effective degree of a polynomial sum $\sum_{\y\in\FO{n}}b_\y\mono{\fvar}{y}{n}$ be the maximum effective degree of any term obtained by expanding each coefficient $b_\y$ to its binary representation, $$b_\y=(b_\y)_02^0 + (b_\y)_12^1 + (b_\y)_22^2.$$ As we prove below, the phase polynomial associated with a tuple $\a\in\Z_8^{2^n-1}$ necessarily evaluates to a non-zero value mod $2^k$ for some input if $\a=\sum_{\y\in\FO{n}}b_\y\mono{\fvar}{y}{n}$ has effective degree $n-k$.

\begin{lemma}\label{lem:idk}
Let $\a\in\Z_8^{2^n-1}$ have effective degree $n-k$ in the monomial basis. There exists $\x\in\F^n$ such that $$\poly_\a(\x) \neq 0 \mod 2^k.$$
\end{lemma}
\begin{proof}
Suppose to the contrary that $\poly_\a(\x) = 0\mod 2^{k}$ for all $\x\in\F^n$. Consider a term $2^i\cdot\mono{\fvar}{y}{n}$ having effective degree $n-l$ for some $l$ greater than or equal to $k$, hence the effective degree is at most $n-k$. By Lemmas \ref{lem:eval} and \ref{lem:weight}, $$\poly_{2^i\cdot\mono{\ev{\fvar}}{y}{n}}(\x)=2^{l}\text{ or } 2^{l - 1},$$ and since $l\geq k$, the result is non-zero mod $2^k$ if and only if $l=k$ (i.e., the effective degree is $n-k$) and $\poly_{2^i\cdot\mono{\ev{\fvar}}{y}{n}}(\x)=2^{l - 1}$. From Lemma \ref{lem:weight} we know that this is the case exactly if $$\deg((\mono{\fvar}{y}{n})(\xor{x}{\fvar}{n}))=\deg(\mono{\fvar}{y}{n})+1,$$ or equivalently the sum $\xor{x}{\fvar}{n}$ contains a non-zero multiple of some variable \emph{not} present in the monomial $\mono{\fvar}{y}{n}$. It can then be observed that for any input $\x\in\F^n$, if $\poly_\a(\x)=0\mod 2^k$, there must be an even number of terms with effective degree $n-k$ that do not contain some variable in the sum $\xor{x}{\fvar}{n}$. We show that this is impossible for all $\x\in\F^n$ by an inclusion--exclusion argument.

We define $S_i$ to be the set of all effective degree $n-k$ terms of $\a$ that do not contain the variable $\fvar_i$. Clearly $\cup_{i|x_i=1} S_i$ gives the set of all such terms that do not contain some variable in the sum $\xor{x}{\fvar}{n}$. Moreover, $|\cup_{i|x_i=1} S_i|$ gives the number of such terms. By the assumption and the observation above that $\poly_\a(\x)=0\mod 2^k$ if and only if there are an even number of terms in $\a$ with effective degree $n-k$ that do not contain some variable in the sum $\xor{x}{\fvar}{n}$, it follows that for any $\x\in\F^n$, $$|\cup_{\fvar_i\in A} S_i| = 0\mod 2.$$

Now take some term $2^i\cdot\mono{\fvar}{y}{n}$ of effective degree $n-k$ in $\a$ but with minimal (actual) degree -- that is, the degree of the monomial part, $\mono{\fvar}{y}{n}$. Since $\mono{\fvar}{y}{n}$ has minimal degree, every other term of effective degree $n-k$ must contain some variable $\fvar_i$ for which $y_i=0$. Hence, the intersection of $S_i$ over all indexes for which $\fvar_i$ is not in the monomial $\mono{\fvar}{y}{n}$ contains exactly one term, $$\cap_{i|y_i=0} S_i = \{2^i\cdot\mono{\fvar}{y}{n}\}.$$ Now, $|\cap_{y_i=0} S_i|$ can be written as a sum of cardinalities of unions of the sets $S_i$: $$|\cap_{i|y_i=0} S_i|=\sum_{\x\in\F^n} s_\x|\cup_{i|x_i=1} S_i|$$ for some integers $s_\x$. For instance, $|S_1\cap S_2| = |S_1| + |S_2| - |S_1\cup S_2|$. However, since $|\cup_{i|x_i=1} S_i| = 0\mod 2$ for any $\x$, we have $|\cap_{i|y_i=0} S_i|=0\mod 2$, a contradiction. Thus there exists $\x\in\F^n$ such that $\poly_\a(\x) \neq 0 \mod 2^k$.
\end{proof}

Lemma~\ref{lem:idk} suffices to prove that $\C$ is generated by the set of scaled monomial vectors of effective degree at most $n-4$. As a consequence we obtain not only a $T$-count optimization procedure for $\{\cnot, T\}$ circuits, but also fully characterize the set of diagonal unitaries implementable over $\{\cnot, T\}$, in the sense that $$\P_8(n)\simeq \Z_8^{2^n-1}/\C.$$

\noindent\textbf{Lemma~\ref{lem:genset}.} \emph{$\C$ is generated by $\{2^i\mono{\ev{\fvar}}{\y}{n} \mid \y\in\F^n, i\in\Z_3, \wt{\y} - i \leq n-4\}$.}
\begin{proof}
Suppose $\c\in\C$. Then $\poly_\c(\x) = 0\mod 8$ for all $\x\in\F^n$, hence by Lemma~\ref{lem:idk}, $\c$ must have effective degree at most $n-4$ and can be written as a sum of the above generators.

Now consider some generator $\c=2^i\mono{\ev{\fvar}}{y}{n}$ where $\wt{\y} - i \leq n-4$. By Lemmas \ref{lem:eval} and \ref{lem:weight}, $$\poly_\c(\x) = 2^{i+n-\wt{\y}}$$ for any $\x\in\F^n$. Since $i+n-\wt{\y} \geq 4$ we have $\poly_\c(\x) = 0\mod 8$ so $\c\in\C$. Moreover since $\C$ is a group, every sum of terms $2^i\mono{\ev{\fvar}}{y}{n}$ where $\wt{\y} - i \leq n-4$ is contained in $\C$, hence $\C$ is generated by $\{2^i\mono{\ev{\fvar}}{\y}{n} \mid \y\in\F^n, \wt{\y} - i \leq n-4\}.$
\end{proof}

\noindent\textbf{Lemma~\ref{lem:general}.} \emph{$\C^{2^k}$ is generated by $\{2^i\mono{\ev{\fvar}}{\y}{n} \mid \y\in\F^n, i\in\Z_3, \wt{\y} - i \leq n-k-1\}$.}
\begin{proof}
Similar to Lemma~\ref{lem:genset}.
\end{proof}

\end{document}